\newcommand{\CLASSINPUTbaselinestretch}{1.2}
\documentclass[12pt,draft,onecolumn]{IEEEtran}

\usepackage{geometry}
\usepackage[english]{babel}

\usepackage[latin1]{inputenc}
\usepackage{enumerate}
\usepackage{color}
\usepackage[T1]{fontenc}
\usepackage{subfigure}
\usepackage{dsfont}
\usepackage[final]{graphicx}
\usepackage[T1]{fontenc}
\usepackage{amsmath}
\usepackage{mathtools}
\usepackage{amsthm}
\usepackage{amstext}
\usepackage{amssymb}
\usepackage{mathrsfs}
\usepackage{cite}
\usepackage{mathtools}
\usepackage{tikz}
\usetikzlibrary{arrows,positioning, shapes}
\usepackage{float}

\usepackage{setspace}
\singlespace

\def\R{\mathbb{R}}

\def\eps{\epsilon}

\def\E{{\mathbb E}}

\def\H{{\mathcal H}}

\def\G{{\mathcal G}}

\def\X{{\mathcal X}}
\def\Y{{\mathcal Y}}

\def\sBer{{\mathsf{Bernoulli}}}

\usepackage[font=small,labelsep=space]{caption}
\DeclareCaptionLabelSeparator{dot}{.~}
\newtheorem{definition}{Definition}
\newtheorem{theorem}{Theorem}
\newtheorem{corollary}{Corollary}

\newtheorem{lemma}{Lemma}
\theoremstyle{remark}
\newtheorem{remark}{Remark}

\usepackage{times}
\usepackage{tikz}
\usepackage{amsmath}
\usepackage{verbatim}
\usetikzlibrary{arrows,shapes}
\tikzstyle{RectObject}=[rectangle,fill=white,draw,line width=0.2mm]
\tikzstyle{line}=[draw]
\tikzstyle{arrow}=[draw, -latex]
\usetikzlibrary{decorations.pathmorphing}
\usetikzlibrary{calc,shapes, positioning}
\usepackage{graphicx}
\usepackage{caption}
\usetikzlibrary{shapes.geometric}

\IEEEoverridecommandlockouts

\allowdisplaybreaks

\begin{document}

\title{\vspace{5.5mm}On Maximal Correlation, Mutual Information\\ and Data Privacy}
\author{\IEEEauthorblockN{Shahab Asoodeh\thanks{This work was supported in part by NSERC of Canada.}, Fady Alajaji, and Tam\'{a}s Linder}\\
    \IEEEauthorblockA{Department of Mathematics and Statistics, Queen's University
    \\\{asoodehshahab, fady, linder\}@mast.queensu.ca}
}
\restoregeometry
\maketitle

\begin{abstract}
The rate-privacy function is defined in \cite{Asoodeh} as a tradeoff between privacy and utility in a distributed private data system in which both privacy and utility are measured using mutual information. Here, we use maximal correlation in lieu of mutual information in the privacy constraint. We first obtain some general properties and bounds for maximal correlation and then modify the rate-privacy function to account for the privacy-constrained estimation problem. We find a bound for the utility in this problem when the maximal correlation privacy is set to some threshold $\epsilon>0$ and construct an explicit privacy scheme which achieves this bound.

%\boldmath

\end{abstract}
\section{Introduction}
 For a given pair of random variables $(X,Y)\in\X\times \Y$, the problem of privacy is, in general, to display a random variable, say $Z$, such that $Y$ and $Z$ are as much correlated as possible while $X$ and $Z$ are almost independent. To make this statement precise, we need to introduce two \emph{measures of dependence}, one for measuring the correlation between $Y$ and $Z$ and the other one between $X$ and $Z$. For two arbitrary alphabets $\mathcal{U}$ and $\mathcal{V}$ and random variables $U\in \mathcal{U}$ and $V\in\mathcal{V}$, a mapping $\delta:\mathcal{U}\times \mathcal{V}\to [0,1]$ is said to be a measure of dependence if $\delta(U;V)=0$ if and only if $U$ and $V$ are independent and $\delta(U;V)=1$ if there exists some deterministic functional relationship between $U$ and $V$, i.e., there exist functions $f$ and $g$ such that $X=f(Y)$ or $Y=g(X)$ with probability one. R\'{e}nyi \cite{Renyi-dependence-measure} postulated additional axioms for an appropriate measure of dependence. For example, the linear correlation coefficient, $|\rho(U;V)|$, is not a measure of dependence as it might become zero even if $U$ is perfectly determined by $V$.

R\'{e}nyi \cite{Renyi-dependence-measure} augmented the definition of the linear correlation coefficient by taking into account functions of $U$ and $V$ and then taking the supremum of $\rho(f(U);g(V))$ over all choices of appropriate functions $f$ and $g$, to obtain \emph{maximal correlation}. There are a few alternative characterizations of maximal correlation in the literature some of which are explained in the sequel. Due to its \emph{tensorization}\footnote{The measure of dependence $\delta(U;V)$ is said to have the tensorization property if for any $n$ i.i.d. copies $(U^n, V^n)$ of $(U,V)$, we have $\delta(U^n;V^n)=\delta(U;V)$. Note that mutual information violates this property as $I(U^n;V^n)=nI(U;V)$.} property, maximal correlation is shown to be very important in correlation distillation, e.g, \cite{Witsenhausen:dependent}, distributions simulation, e.g., \cite{Gohari-additivity}, and is also related to the hypercontractivity coefficient, e.g., \cite{anantharam} and \cite{Sudeep-Kamath-non-interactive}. Beigi and Gohari \cite{Gohari-nonlocal-Correlation} have recently proposed \emph{maximal correlation ribbon} as a generalization of maximal correlation.

Mutual information $I(U;V)$ can also be viewed as a measure which captures dependence between $U$ and $V$. Although, it is not a measure of dependence according to R\'{e}nyi's stipulations, it has some properties which make mutual information a good candidate in data privacy applications especially for measuring  utility.  Although both maximal correlation and mutual information have been used in numerous applications in information theory, the connection between them is still not fully explored in the literature.

%In this work, we obtain some general results on maximal correlation, explore the connection between maximal correlation and mutual information, and then formulate a data privacy problem in terms of maximal correlation. We then derive some achievability results for the utility.
The definition of maximal correlation together with some alternative characterizations are given in Section II. In Section III, we present some general results about maximal correlation and also some bounds in terms of mutual information. We then formulate a data privacy problem (privacy-constrained estimation) in terms of maximal correlation in Section IV and present some achievability results.

\section{Maximal Correlation: Definition and Characterization}
Suppose that $X$ is a random variable with distribution $P$, over alphabet $\X$ and $Y$ is another random variable which results from passing $X$ through channel $W$. Channel $W$ consists of a family of probability measures defined over alphabet $\Y$, i.e., $P_{Y|X}(\cdot|x)$ for $x\in\X$. We denote by $W\circ P$ the distribution on $\Y$ induced by the push-forward of the distribution $P$, which is the distribution of the output  $\Y$ when the input  $X$ is distributed according to $P$, and by $P\times W$ the joint distribution $P_{XY}$ if $P_X=P$.

Let $\G$ (resp. $\H$) be the set of all real-valued functions of $X$ (resp. $Y$) with zero mean and finite variances with respect to $P$ (resp.  $W\circ P$). The sets $\G$ and $\H$ are both separable Hilbert spaces with the covariance as the inner product.

For a fixed channel, $W$, the maximal correlation between $X$ and $Y$ is a functional of $P$ and $W$ defined as
\begin{eqnarray}\label{def_maximal}
    \rho_m(P;W)&:=& \sup_{g\in \G, f\in \H} \rho(g(X); f(Y))\\
    &=& \sup_{g\in \G, f\in \H, ||f||_2=||g||_2=1} \mathbb{E}[g(X)f(Y)]\nonumber,
\end{eqnarray}
where $\rho(\cdot;\cdot)$ is the linear correlation coefficient\footnote{I.e., $\rho(X;Y):=\frac{\text{Cov}(X;Y)}{\sigma_X\sigma_Y}$, where $\text{Cov}(X;Y),\sigma_X$ and $\sigma_Y$ are the covariance between $X$ and $Y$, the standard deviation of $X$ and the standard deviation of $Y$, respectively.} and for any random variable $U$, $||U||_2^2:=\mathbb{E}[U^2]$.
We use  interchangeably the notation $\rho_m(P; W)$ and $\rho_m(X;Y)$ where $X\sim P$ and $Y$ are respectively the input and output of channel $W$. Maximal correlation is a measure of dependence between random variables $X$ and $Y$, that is, $0\leq \rho_m(X;Y)\leq 1$ where $\rho_m(X;Y)=0$ if and only if $X$ and $Y$ are independent and $\rho_m(X;Y)=1$ if an only if there exists a pair of functions $g$ and $f$ such that $g(X)$ and $f(Y)$ are non-degenerate and $f(Y)=g(X)$ with probability one.
Maximal correlation is closely related to the conditional expectation operator, defined as follows.
\begin{definition}
For a given joint distribution $P_{XY}=P\times W$, the conditional expectation operator $T_X:\H\to \G$ is defined as
$$(T_Xf) (x):=\mathbb{E}[f(Y)|X=x].$$
\end{definition}
It is a well-known fact that the second largest singular value\footnote{For any arbitrary operator $T$ mapping (Banach) space $\X$ to itself, an eigenvalue of $T$ is defined as a number $\lambda$ such that $Tx=\lambda x$. The singular value of $T$ is then defined as the eigenvalue of $T^*T$ where $T^*$ is the adjoint of $T$. See \cite{invitation-abramovich2002} for more details.} of $T_X$ is precisely $\rho_m(P;W)$, see e.g.  \cite{Witsenhausen:dependent} and \cite{Renyi-dependence-measure}.
%For any pair of functions $f\in\H$ and $g\in\G$, we hence have \begin{equation}\label{functions-norm-product}
%    \mathbb{E}[g(X)f(Y)]=\langle T_Xg,f \rangle \le \rho_m(P;W)||g||_2||f||_2.
%\end{equation}

The definition of maximal correlation, given in \eqref{def_maximal}, has been simplified in the literature in general and also for some special cases. For example, by a simple application of the Cauchy-Schwarz inequality, R\'{e}nyi \cite{Renyi-dependence-measure} showed the following one-function characterization,
\begin{equation}\label{one-function-def}
    \rho^2_m(P;W)=\sup_{g\in\G, ||g||_2=1}\mathbb{E}[\mathbb{E}^2[g(X)|Y]].
\end{equation}
\begin{remark} If $\min\{|\X|, |\Y|\}=2$, then
\begin{equation}\label{binary-maximal}
    \rho_m^2(P;W)=\chi^2(P_{XY}||P_X\times P_Y),
\end{equation}
where the chi-squared divergence is defined as
\begin{equation}\label{cho-squared-def}
    \chi^2(P||Q):=\int \left(\frac{\text{d}P}{\text{d}Q}\right)^2\text{d}Q-1,
\end{equation}
where $\frac{\text{d}P}{\text{d}Q}$ is the Radon-Nikodym derivative of $P$ with respect to $Q$. Note that in the finite dimensional case, the singular values of operator $T_X$ are equal to the singular values of the matrix $Q=[\frac{P_{XY}(x,y)}{\sqrt{P_X(x)P_Y(y)}}]$, see \cite{linearinfo2}. The expression \eqref{binary-maximal} therefore follows from observing that $\rho_m^2(P;W)$ is the second largest eigenvalue of both $QQ^T$ and $Q^TQ$ either of which is a $2\times 2$ matrix which implies that $\rho_m^2(P;W)$ is equal to the trace of that matrix minus the largest eigenvalue, i.e., 1. It is important to mention here that $\chi^2(P_{XY}||P_XP_Y)$ is shown in \cite{Witsenhausen:dependent} to be equal to the sum of squares of the singular values of operator\footnote{In the finite dimensional case, the sum of the singular values of operator $T$ is equal to the Frobenius norm of $T$ which is defined as  $||T||_F=\text{Tr}(T^*T)$ where $\text{Tr}$ is the trace operator.} $T_X$ minus one (i.e., the largest one) while $\rho_m(X;Y)$ is the second largest one.
\end{remark}
Suppose $\tilde{W}$ is the backward channel corresponding to $W$, that is, if $W=P_{Y|X}$, then $\tilde{W}=P_{X|Y}$. Then the composition $\tilde{W}\circ W:\X\to \X$ defined by
$$\tilde{W}\circ W(x'|x)=\sum_{y\in\Y}W(y|x)\tilde{W}(x'|y),$$
is a channel for which $P$ is a stationary distribution and the associated conditional expectation operator $T_X$ is self-adjoint. It is easy to show that in this case
\begin{equation}\label{backward-channel}
    \rho_m^2(P;W)=\rho_m(P;\tilde{W}\circ W).
\end{equation}
To see this, note that it is show in \cite{anantharam} that
\begin{equation}\label{direct-inequality}
    \rho_m^2(P;W)=\sup_{g\in\G, ||g||_2=1}\mathbb{E}[g(X)g(X')],
\end{equation}
where $X'$ is the output of channel $\tilde{W}\circ W$ under input $X$. This clearly implies that $\rho_m^2(X;Y)\leq \rho_m(X;X')$. The following gives the reverse inequality. For arbitrary measurable functions $h, g:\X\to \R$, we have
\begin{eqnarray} \label{chain-of-inequa}
  \mathbb{E}[g(X)h(X')] &\stackrel{(a)}{=}& \mathbb{E}\Big[\mathbb{E}[g(X)|Y]\mathbb{E}[h(X')|Y]\Big] \nonumber\\
   &\stackrel{(b)}{\leq}& ||\mathbb{E}[g(X)|Y]||_2||\mathbb{E}[h(X')|Y]||_2\nonumber\\
   &\stackrel{(c)}{\leq }& \rho_m(X;Y)\rho_m(X'; Y)\nonumber\\
    &\stackrel{(d)}{=}& \rho_m^2(X;Y), \end{eqnarray}
where $(a)$ is due to the Markov condition $X\to Y\to X'$,  $(b)$ is a simple application of the Cauchy-Schwarz inequality, $(c)$ comes from \eqref{one-function-def}, and $(d)$ follows from the fact that $\rho_m(X'; Y)=\rho_m(X; Y)$. This chain of inequalities shows that $\rho_m(X;X')\leq \rho_m^2(X;Y)$ which, together with the earlier inequality, yields $\rho_m(X;X')=\rho_m^2(X;Y)$.
%It is a natural question that for what pair of channels, the equation $\rho_m(P;W)=\rho_m(P; \tilde{W})$ holds? and also what does this imply in an information theoretic context. It is reminiscent of the preserving of mutual information; $I(X;Y)=I(X;Z)$, which is equivalent to saying that $Z$ is a "\emph{sufficient statistic}" of $Y$ with respect to $X$, see \cite{Cover:2006}.
\section{Maximal Correlation and Mutual Information}
It is well-known that for Gaussian random variables $X$, $Y$ and $Z$ which satisfy the Markov condition  $X\to Y\to Z$, we have $\rho(X,Z)=\rho(Y, Z)\rho(X,Y)$. A similar relation for maximal correlation does not in general hold. However, the following theorem gives a similar result.
\begin{theorem}
For random variables $X$ and $Y$ with a joint distribution $P\times W$, we have
$$\sup_{\substack{X\to Y\to Z\\ \rho_m(Y;Z)\neq 0}}\frac{\rho_m(X;Z)}{\rho_m(Y;Z)}=\rho_m(X;Y).$$
\end{theorem}
\begin{proof}
First note that by data processing for maximal correlation the ratio on the left-hand side is always less than or equal to one. The inequality (c) in \eqref{chain-of-inequa} yields $\rho_m(X; Z)\leq \rho_m(X;Y)\rho_m(Y;Z)$ from which we can write
$$\frac{\rho_m(X; Z)}{\rho_m(Y;Z)}\leq \rho_m(X;Y).$$
The achievability result comes from the special case treated in Section II where $X\to Y\to X'$ and $P_{X'|Y}$ is the backward channel associated with $P_{Y|X}$. It was shown that $\rho_m(X;Y)\rho_m(X';Y)=\rho_m(X;X')$ which completes the proof.
\end{proof}
This theorem is similar to a recent result by Anantharam et al. \cite{anantharam} in which for a given $P_{XY}$ the ratio between $I(X;Z)$ and $I(Y;Z)$ is maximized over all channels $P_{Z|Y}$ such that the Markov condition $X\to Y\to Z$ is satisfied.

The following theorem connects the maximal correlation with mutual information when $X$ and channel $W$ are both assumed to be Gaussian.
\begin{theorem}\label{theorem-Gaussian}
Let $(X,Y)$ be jointly Gaussian random variables, then we have
$$\rho_m^2(X;Y)\leq 1-2^{-2I(X;Y)}\leq (2\ln 2)I(X; Y).$$
\end{theorem}
\begin{remark}
Linfoot \cite{Linfoot195785} introduced the \emph{informational} measure of correlation which is defined for two continuous random variables $X$ and $Y$ as $$r(X;Y):=\sqrt{1-2^{-2I(X;Y)}}.$$ Theorem~\ref{theorem-Gaussian} therefore implies that for jointly Gaussian random variables, $\rho_m(X;Y)\leq r(X;Y)$. The informational measure of correlation is generalized in \cite{Thesis-Shan-Lu} for general random variables.
\end{remark} \vspace{-0.2cm}
\begin{proof}
Since $(X,Y)$ is bivariate Gaussian, we know from \cite{lancaster} that  $\rho_m(X;Y)=|\rho(X;Y)|$. On the other hand, we can show that given a pair of random variables $X$ and $Y$, the conditional expectation of $X$ given $Y$ has the maximum linear correlation with $X$ among all functions $f\in\H$, i.e.
\begin{equation}\label{correlation-supremum}
    \sup_{f}\rho(X; f(Y))=\rho(X; \mathbb{E}[X|Y])=\frac{||\mathbb{E}[X]-\mathbb{E}[X|Y]||_2}{\sqrt{\mathsf{var}(X)}},
\end{equation}
where the supremum is taken over all measurable functions $f$ with finite variance (not necessarily with zero mean) and $\mathsf{var}(X)$ denotes the variance of $X$. To see this, without loss of generality, we can assume that $f\in \H$, i.e., $\mathbb{E}[f(Y)]=0$. Then we have
\begin{eqnarray*}
  \rho(X; f(Y)) &=& \frac{\mathbb{E}[X f(Y)]}{\sqrt{\mathsf{var}(X)}||f(Y)||_2} \\
   &=&  \frac{\mathbb{E}\big[f(Y)\mathbb{E}[X|Y]\big]}{\sqrt{\mathsf{var}(X)}||f(Y)||_2}\leq \frac{||\mathbb{E}[X|Y]||_2}{\sqrt{\mathsf{var}(X)}},
  % &\leq & \frac{||\mathbb{E}[X|Y]||_2}{\sqrt{\mathsf{var}(X)}},
\end{eqnarray*}
where the inequality comes from the Cauchy-Schwarz inequality. Equality occurs if $f(Y)=\mathbb{E}[X|Y]$. It is a well-known fact from rate-distortion theory that  for Gaussian $X$ and its reconstruction $\hat{X}$
$$I(X; \hat{X})\geq \frac{1}{2}\log\frac{\mathsf{var}(X)}{\mathbb{E}[(X-\hat{X})^2]},$$ and hence by setting $\hat{X}=\mathbb{E}[X|Y]$, after some straightforward calculations we obtain \begin{equation}\label{mutualinfo}
    I(X; Y)\geq \frac{1}{2}\log\frac{1}{1-\rho^2(X; \mathbb{E}[X|Y])},
\end{equation}
 and hence,
 \begin{equation}\label{mutualinfo2}
    \rho^2(X; \mathbb{E}[X|Y])\leq 1-2^{-2I(X;Y)}.
\end{equation}
Combining \eqref{correlation-supremum} and \eqref{mutualinfo2}, we have
\begin{eqnarray*}
  \rho_m^2(X;Y) &\leq& \rho^2(X; \mathbb{E}[X|Y])\leq  1-2^{-2I(X;Y)} \\
   &= & 1-e^{-2\ln2I(X;Y)}\leq 2\ln2I(X;Y).
\end{eqnarray*} \vspace{-0.3cm} \end{proof}
Note that Theorem~\ref{theorem-Gaussian} is based on the fact that for jointly Gaussian random variables $X$ and $Y$, we have $\rho_m(X;Y)=|\rho(X; Y)|$. This is not, in general, true. For example consider a pair of zero-mean random variables $X=U_1V$ and $Y=U_2V$ where all $U_1$, $U_2$ and $V$ are independent and $\Pr(U_i=+1)=\Pr(U_i=-1)=1/2$ for $i=1,2$. We have $\mathbb{E}[X|Y]=\mathbb{E}[U_1V|U_2V]=0$ and similarly $\mathbb{E}[Y|X]=0$ both implying that $\rho(X;Y)=0$. Nevertheless, $\Pr(X^2=Y^2)=1$ implying that $\rho_m(X;Y)=1$.
%Before the theorem we need the following lemma which gives a bound for the gap of Jensen's inequality.
%\begin{lemma}{\cite{Jensen'sGap}}\label{jensen'sgapLemma}
%For any non-negative random variable $X$, we have
%$$0\leq \mathbb{E}[\exp(-X)]-\exp(-\mathbb{E}[X])\leq \mathsf{var}(X).$$
%\end{lemma}

The following theorem gives a lower bound for maximal correlation in terms of mutual information. We assume that the Radon-Nikodym derivative $P_{XY}$ with respect to $P_X\times P_Y$ exists which we denote it by $\imath$, i.e.,
\begin{equation}\label{information-density}
    \imath:=\frac{\text{d}P_{XY}}{\text{d}(P_X\times P_Y)}.
\end{equation}
The logarithm of this quantity is sometimes called the information density \cite[p. 248]{Gray:1990:EIT:90455}.
\begin{theorem}\label{theorem-min=2}
For a given $P_{XY}=P\times W$ with $\min\{|\X|, |\Y|\}=2$, we have
$$\rho_m^2(P; W)\geq 2^{I(P;W)}-1$$
%$$2^{I(P;W)}-1\leq \rho_m(P; W)\leq 2^{I(P;W)}+\mathsf{var}(\ell(X,Y))-1.$$
\end{theorem}
\begin{proof}
As mentioned earlier, when $\min\{|\X|, |\Y|\}=2$, then $\rho_m^2(X;Y)=\chi^2(P_{XY}||P_XP_Y)$ and hence
\begin{eqnarray}
  \rho_m^2(X;Y) &=& \int \text{d}P_{XY}\left(\frac{\text{d}P_{XY}}{\text{d}(P\times P_Y)}\right)-1 \nonumber\\
   &=& \mathbb{E}_{P_{XY}}\Big[2^{\log \imath(X,Y)}\Big]-1\nonumber\\
   &\geq & 2^{\mathbb{E}_{P_{XY}}[\log \imath(X,Y)]}-1,
      %\int \text{d}P_{XY}2^{\log \frac{\text{d}P_{XY}}{\text{d}(P\times P_Y)}}-1\nonumber
\end{eqnarray}
where the inequality is due to Jensen's inequality.
%For the upper bound, we can use Lemma~\ref{jensen'sgapLemma} to bound the gap of Jensen's inequality that we use for the lower bound. We can write
%\begin{eqnarray*}
%% \nonumber to remove numbering (before each equation)
%  \rho_m(X;Y) &=&  \mathbb{E}_{P_{XY}}\big[2^{-\jmath(X,Y)}\big]-1\\
%   &=& \mathbb{E}_{P_{XY}}\Big[2^{-\jmath(X,Y)}1_{\{\jmath(X,Y)\geq 0\}}\Big]+\\
%   && \mathbb{E}_{P_{XY}}\Big[2^{-\jmath(X,Y)}1_{\{\jmath(X,Y)< 0\}}\Big] -1
%%   &=&
%\end{eqnarray*}
\end{proof}
Theorem~\ref{theorem-min=2} hods only when either $|\X|=2$ or $|\Y|=2$. Suppose we have a binary-input AWGN channel modeled by $Y=X+N$, where $X\sim\sBer(p)$ and $N\sim\mathcal{N}(0, \sigma^2)$ are independent. Theorem~\ref{theorem-min=2} then implies that if $I(X;Y)\to 1$ (which occurs only when $\sigma^2\to 0$) then there exists a pair of functions $f\in\H$ and $g\in\G$ such that $f(Y)=g(X)$ with probability one. The following theorem gives an upper bound for maximal correlation when $|\X|<\infty$.
\begin{theorem}\label{theorem-Min-PX}
If $X$ is a discrete random variable with $|\X|<\infty$, then for a given joint distribution $P_{XY}=P\times W$, we have
$$P_{min}\rho_m^2(P;W)\leq \sqrt{(2\ln 2) I(P;W)},$$ where $P_{min}:=\min_{x\in\X} P(x)$.
\end{theorem}
\begin{proof}
In the proof we assume that $Y$ has also a finite alphabet, however, the proof can be modified for general alphabet $\Y$. As mentioned earlier, for any pair of random variables $(X,Y)$, $\rho^2_m(X;Y)\leq \chi^2(P_{XY}||P\times P_Y)$ and hence
%
%
%and $\chi^2(P_{XY}||P\times P_Y)$ are equal to the second largest and the sum of the singular values of the operator $T_X$ minus the largest one, respectively, and hence
\begin{eqnarray*}
  \rho^2_m(X;Y) &\leq & \chi^2(P_{XY}||P\times P_Y) \\
  % &=& \sum_{x\in\X, y\in\Y}P_{XY}(x,y)\frac{P_{XY}(x,y)}{P(x)\times P_{Y}(y)}-1 \\
   &=&  \sum_{x,y}\left(P_{XY}(x,y)-P(x) P_Y(y)\right)\frac{P_{XY}(x,y)}{P(x) P_{Y}(y)}\\
   &\leq & \max_{x,y}\frac{P_{XY}(x,y)}{P(x)\times P_{Y}(y)}||P_{XY}-P\times P_Y||_{TV}\\
   &\leq & \frac{1}{P_{min}}||P_{XY}-P\times P_Y||_{TV}\\
   &\leq & \frac{1}{P_{min}} \sqrt{(2\ln 2) I(P;W)},
\end{eqnarray*}
where $||Q-P||_{TV}:=\sum_{x}|Q(x)-P(x)|$ is the total variation distance for probability measures $Q$ and $P$ and the last inequality is due to Pinsker's inequality (see e.g., \cite[problem 3.18]{csiszarbook}).
\end{proof}
The value of the maximal correlation is often hard to calculate except for a few classes of joint distributions. For instance, as mentioned earlier, if $(X,Y)$ is jointly Gaussian then  the exact value of $\rho_m(X;Y)$ is known. Bryc et al. \cite{dembo1} showed that there exists another family of joint distributions for which the maximal correlation can be exactly computed. For this, we need the following definition.
\begin{definition}\hspace{-0.1cm}\cite{Durret'sBook}
A random variable $X$ is said to have an $\alpha$-stable distribution if the characteristic function of $X$ is of the form
\begin{eqnarray*}
  \varphi(t) &:=&\mathbb{E}[\exp(itX)]  \\
   &=& \exp\left(itc-b|t|^{\alpha}(1+i\kappa ~\text{sgn}(t)\omega_{\alpha}(t))\right),
\end{eqnarray*}
where $c$ is a constant, $\text{sgn}$ is the sign function, $-1\leq \kappa\leq +1$ and
\begin{equation*}
\omega_{\alpha}(t) = \begin{cases}
\tan (\frac{\pi\alpha}{2}) &\text{if  $\alpha\neq 1$}\\
\frac{2}{\pi}\log|t| &\text{if $\alpha=1$. }
\end{cases}
\end{equation*}
\end{definition}
Gaussian, Cauchy and L\'{e}vy distributions are examples of stable distributions.
\begin{theorem}\hspace{-0.3cm} \cite{dembo1}\label{stable-Theorem}
Let $(X, Y)$ be a given pair of random variables. \\
(I). If $N$ is a random variable with an $\alpha$-stable distribution and is independent of $(X,Y)$, then $\lambda\mapsto \rho_m(Y; X+\lambda N)$ is a non-increasing function for $\lambda\ge 0$.\\
(II). If $N$ and $X$ are independent and have the same $\alpha$-stable distribution for $0<\alpha\leq 2$, then for any $\lambda\geq 0$, $$\rho_m(X, X+\lambda N)=\frac{1}{\sqrt{1+\lambda^{\alpha}}}.$$
\end{theorem}
This theorem shows that if $W$ (the channel $\X\to \Y$) is an additive noise channel, $Z=X+\lambda N$, where $N$ and $X$ have an $\alpha$-stable distribution, then $\rho_m(X;Z)$ can be analytically calculated. Part (I) of this theorem might look trivial at first, as for $N$ independent of $(X,Y)$, one might think that $Y$ and $X+\lambda N$ are asymptotically independent when $\lambda\to \infty$. However this does not, in general, hold. For example let $X$ take value in $[0, 1]$ and $N$ be a binary random variable taking values $+1$ and $-1$. Then $X+N$ is mapped either to $[1,2]$ or $[-1, 0]$ which are two disjoint sets and hence for any known $|\lambda|>1$, $X+\lambda N$ determines uniquely the value of $X$.
\section{A problem of privacy}
The tradeoff between data privacy and utility has always been an intriguing problem in computer science and information theory. Information-theoretic  privacy was first studied by Shannon who connected information theory to cryptography. Yamamoto \cite{yamamotoequivocationdistortion} introduced a set-up where given $n$ i.i.d. copies of two correlated sources $X$ and $Y$, the receiver is to be able to reconstruct $Y$ within a distortion $D$ and unable to estimate $X$, and hence $X$  is kept private from the receiver. In this set-up privacy is measured in terms of \emph{equivocation} which is the conditional entropy of $X$ given what the receiver observes. Yamamoto \cite{yamamotoequivocationdistortion} characterized the tradeoff between distortion and equivocation. Another set-up for privacy is given in \cite{Asoodeh} where both utility and privacy are defined in terms of mutual information and the \emph{rate-privacy function} is introduced as the tradeoff between utility and privacy.
\begin{definition}
For a given joint distribution $P\times W$, the rate-privacy function is defined as
$$g_{\epsilon}(P, W):=\sup \{I(Y; Z):~ X\to Y\to Z, ~I(X;Z)= \epsilon\}.$$
\end{definition}
The channel $P_{Z|Y}$, over which the supremum is taken, is in fact responsible for masking information about $X$ and is thus called a privacy filter. Thus, $g_{\epsilon}(P, W)$ quantifies the maximum information that one can  receive about $Y$ while revealing only $\epsilon$ bits of information about $X$. From the privacy point of view, the case with zero privacy leakage is of more interest, i.e., $\epsilon=0$, which is called \emph{perfect privacy}. It is shown in \cite{Asoodeh} that for finite $\X$ and $\Y$,  $g_{0}>0$ if and only if vectors $\{P_{X|Y}(\cdot|y):~y\in\Y\}$ are linearly dependent implying that the matrix corresponding to joint distribution $P_{XY}$ is rank-deficient. In particular if $|\Y|>|\X|$, then $g_{0}>0$.

The following lemma shows that the mapping $\epsilon\mapsto \frac{g_{\epsilon}(P, W)}{\epsilon}$ is non-increasing.
\begin{lemma}\label{non-increasing-Lemma}
For a given joint distribution $P\times W$, $\epsilon\mapsto \frac{g_{\epsilon}(P, W)}{\epsilon}$ is non-increasing.
\end{lemma}
\begin{proof}
The proof follows the same steps as the proof of \cite[Lemma. 1]{schulman}.
\end{proof}
%\begin{proof}
%For a given channel $P_{Z|Y}$ and $\delta>0$, we can define a new channel with an additional symbol $e$ as follows
%\begin{equation*}
%P_{Z'|Y}(z'|y) = \begin{cases}
%(1-\delta)P_{Z|Y}(z'|y) &\text{if  $z'\neq e$}\\
%\delta &\text{if $z'=e$ }
%\end{cases}
%\end{equation*}
%It is easy to check that $I(Y; Z')=(1-\delta)I(Y; Z)$ and also $I(X; Z')=(1-\delta)I(X; Z)$. Now suppose that $P_{Z|Y}$ achieves $g_{\epsilon}(P;W)$, that is, $g_{\epsilon}(P;W)=I(Y; Z)$ and $I(X;Z)=\epsilon$. We can then write
%$$\frac{g_{\epsilon}(P;W)}{\epsilon}=\frac{I(Y;Z)}{\epsilon}=\frac{I(Y; Z')}{\epsilon'}\leq \frac{g_{\epsilon'(P;W)}}{\epsilon'},$$
%where $\epsilon'=(1-\delta)\epsilon$. Therefore, for $\epsilon'< \epsilon$ we have $\frac{g_{\epsilon'}(P;W)}{\epsilon'}\geq \frac{g_{\epsilon}(P;W)}{\epsilon}$.
%\end{proof}
This lemma yields the following bound for the rate-privacy function. \begin{corollary}
For a given joint distribution $P\times W$, we have for any $\epsilon>0$
$$g_{\epsilon}(P,W)\geq \epsilon\frac{H(Y)}{I(P;W)}.$$
\end{corollary}
\begin{proof}
By  the Markov condition $X\to Y\to Z$, we know that $\epsilon\leq I(P;W)$. When $\epsilon=I(P;W)$ then the privacy constraint is removed and hence $g_{I(P;W)}=H(Y)$. The result then follows from Lemma~\ref{non-increasing-Lemma}.
\end{proof}
It is important to note, however, that the mutual information has deficiencies as a measure of privacy (e.g. \cite{Evfimievski:2003:LPB:773153.773174}). We can, instead, use maximal correlation as a measure of privacy and then define
$$\hat{g}_{\epsilon}(P,W):=\sup \{I(Y; Z):~ X\to Y\to Z, ~\rho_m(X;Z)\leq \epsilon\},$$ as the corresponding privacy-rate tradeoff.

Suppose now that the privacy filter is such that the Markov condition $X\to Y\to Z$ is satisfied and the channel $P_{Z|X}$ can be modeled by $Z=X+\lambda N$ for $\lambda>0$ where $N$ and $(X, Y)$ are independent and has the same $\alpha$-stable distribution as $X$ for some $\alpha\in(0,2]$. Then by Theorem~\ref{stable-Theorem}, we know that $\rho_m(X;Z)=\frac{1}{\sqrt{1+\lambda^{\alpha}}}$. Let $$\varrho_{\eps}(X;Y):=\rho_m(Y; X+\lambda^{*}N),$$ where $$\lambda^{*}_{\eps}=\left(\frac{1}{\epsilon^2}-1\right)^{1/\alpha}.$$
We can therefore conclude from Theorem~\ref{stable-Theorem} that
\begin{equation}\label{maximum-of-maximal}
    \max \rho_m(Y;X+\lambda N)=\varrho_{\eps}(X;Y)
\end{equation}
where the maximum is taken over all $\lambda$ such that $\rho_m(X;X+\lambda N)\leq \epsilon$.
This says that if the privacy filter meets the above model, then the best $\lambda$ which satisfies $\epsilon$ maximal correlation privacy; $\rho_m(X;Z)\leq \epsilon$, is $\lambda^{*}_{\eps}$. In other words, among all such privacy filters
\begin{equation}\label{rho_privacy_maximization}
    \sup_{\rho_m(X;Z)\leq \epsilon}\rho_m(Y; Z)=\varrho_{\eps}(X;Y).
\end{equation}
%and by Theorem~\ref{theorem-min=2} if $|\X|=2$
%$$\sup_{\rho_m(X;Z)\leq \epsilon} I(Y;Z)\leq \log\left(\rho^2_{\dagger}(X;Y)+1\right), $$
Unfortunately, all stable distributions have infinite support (like the Poisson and Gaussian distributions), thus $|\Y|=\infty$, and hence we can not invoke Theorem~\ref{theorem-Min-PX} to obtain a lower bound for $\hat{g}_{\epsilon}(P,W)$. Finding a similar upper-bound of $\rho_m(X;Y)$ in terms of mutual information for general alphabets remains open. It is worth mentioning that the channel model, $Z=X+\lambda N$ is similar to the \emph{artificial noise} introduced in \cite{artificial_noise} in which both signal and noise are assumed to be Gaussian, i.e., having a $2$-stable distribution.

Defining a utility in terms of linear correlation coefficient, we can construct a \emph{privacy-constrained estimation problem}. Suppose an agent knowing $Z$ wants, on the one hand to estimate $Y$ as reliably as possible, and on the other hand, to satisfy the privacy constraint $\rho_m(X;Z)\leq \eps$. Let $\mathsf {mmse}(Y; \lambda)$ denote the minimum mean squared error (MMSE) of $Y$ based on $Z=X+\lambda N$, that is $$\mathsf {mmse}(Y; \lambda):=\mathbb{E}\left[\big(Y-\mathbb{E}[Y|X+\lambda N]\big)^2\right].$$ Let $\mathsf{mmse}_{\eps}(Y)$ denote the minimum achievable $\mathsf{mmse}(Y; \lambda)$ when $\rho_m(X;Z)\leq \eps$.
\begin{theorem}
If the privacy filter $P_{Y|Z}$ is such that for random variables $X\to Y\to Z$, $P_{Z|X}$ can be modeled as $Z=X+\lambda N$, for $N$ independent of $(X,Y)$ and having similar $\alpha$-stable distribution as $X$ for $\alpha\in(0, 2]$. Then
$$\mathsf {mmse}_{\eps}(Y)\geq (1-\varrho_{\eps}^2(X;Y))\mathsf{var}(Y).$$
\end{theorem}
\begin{proof}
By simple algebraic manipulations, we can write
\begin{eqnarray*}
  \mathsf {mmse}(Y; \lambda) &=& \E[Y^2]-\E[\E^2[Y|Z]] \\
   &=& \mathsf{var}(Y)-||\E[Y]-\E[Y|Z]||_2^2 \\
   &\stackrel{(a)}{=}&  \mathsf{var}(Y)[1-\rho^2(Y; \E[Y|Z])],
\end{eqnarray*}
where $(a)$ is obtained from \eqref{correlation-supremum}. Since $\rho(Y, g(Z))\leq \rho_m(Y;Z)$ for any function $g$, we have
$$\mathsf {mmse}(Y; \lambda) \geq \mathsf{var}(Y)(1-\rho^2_m(Y; Z)).$$
%\begin{eqnarray*}
%% \nonumber to remove numbering (before each equation)
%  \mathsf {mmse}(Y; \lambda) &\geq& \mathsf{var}(Y)(1-\rho^2_m(Y; Z))\\
%&\stackrel{(b)}{\geq}&  \mathsf{var}(Y)(1-\varrho_{\eps}^2(X; Y)),
%\end{eqnarray*}
The result follows by taking minimum from both sides over $\lambda$ such that $\rho_m(X;Z)\leq \eps$ and invoking \eqref{maximum-of-maximal}.
\end{proof}
The lower bound for MMSE becomes zero only if $\varrho_{\eps}(X; Y)=1$. It is easy to verify  that in the trivial Markov chain $Y\to X\to \lambda^*_{\epsilon} N$, we have $\rho_m(Y;X)\geq \rho_m(Y; X+\lambda^*_{\epsilon} N)$, therefore if $\rho_m(X;Y)<1$, then $\varrho_{\eps}(X;Y)<1$ and thus $(1-\varrho_{\eps}^2(X;Y))$ is bounded away from zero. This is the price that one has to pay to have \emph{privacy-constrained} estimation. We note that $\varrho_{\eps}(X;Y)$ is non-increasing in $\eps$ and thus for a  more stringent privacy constraint (i.e., smaller $\eps$) we have bigger $\mathsf{mmse}_{\eps}(Y)$.
%we have
%$$\sup_{\rho_m(X;Z)\leq \epsilon} I(Y;Z)\geq 2\rho^2_{*}(X;Y)p^2_{min},$$
%where $p_{min}:=\min_{y\in\Y}P_Y(y)$.
%
%This, in general, can serve as a lower bound for $\hat{g}_{\epsilon}(P;W)$:
%$$\hat{g}_{\epsilon}(P;W)\geq \frac{1}{8\ln 2}\rho^2_{*}(X;Y)p^2_{min}.$$
%
%\section{Conclusion}
%In this work, we presented a result on the maximum of the ratio between maximal correlations $\rho_m(X;Z)$ and $\rho_m(Y;Z)$ in a Markov chain $X\to Y\to Z$ for a given joint distribution $P_{XY}$. We then gave upper and lower bounds for maximal correlations in terms of mutual information. Measuring privacy in terms of maximal correlation, we modified the rate-privacy function, as a trade-off between utility to a define a private estimation problem. Having constructed an explicit privacy scheme, we derived a lower bound for MMSE in that private estimation problem.

\bibliographystyle{IEEEtran}
\bibliography{bibliography}

\end{document}